\documentclass[10pt, conference, letterpaper]{IEEEtran}
\usepackage{url}
\usepackage{amssymb}
\usepackage{array}
\usepackage{multirow}
\usepackage{amsmath}
\usepackage{graphicx}
\usepackage[ruled,vlined]{algorithm2e}
\usepackage{color}
\usepackage{caption,subcaption}
\usepackage{threeparttable}
\newtheorem{theorem}{Theorem}
\usepackage{threeparttable}
\usepackage{bbm}
\usepackage{float}
\usepackage{pgfplots}

\newtheorem{lemma}{Lemma}

\ifCLASSINFOpdf
\newcommand{\be}{\begin{equation}}
\newcommand{\ee}{\end{equation}}
\newcommand{\cJ}{\mathcal{J}}
\newcommand{\cV}{\mathcal{V}}
\newcommand{\cU}{\mathcal{U}}
\def\nn{\nonumber}
\def\bea{\begin{eqnarray}}
\def\eea{\end{eqnarray}}
\def\nn{\nonumber}

\def\ba{\begin{array}}
\def\ea{\end{array}}

\def\nn{\nonumber}
\else

\fi

\begin{document}

\title{Online Joint Placement and Allocation of Virtual Network Functions with Heterogeneous Servers}

\author{\IEEEauthorblockN{
Yicheng Xu\IEEEauthorrefmark{1} Vincent Chau\IEEEauthorrefmark{2} Chenchen Wu\IEEEauthorrefmark{3} Yong Zhang\IEEEauthorrefmark{4} Yifei Zou\IEEEauthorrefmark{5}}\\
\IEEEauthorblockA{\IEEEauthorrefmark{1}Shenzhen Institutes of Advanced Technology, Chinese Academy of Sciences, P.R. China~
Email: yc.xu@siat.ac.cn\\
\IEEEauthorrefmark{2}Shenzhen Institutes of Advanced Technology, Chinese Academy of Sciences, P.R. China~
Email: vincentchau@siat.ac.cn\\
\IEEEauthorrefmark{3}College of Science, Tianjin University of Technology, P.R. China~
Email: wu\_chenchen\_tjut@163.com\\
\IEEEauthorrefmark{4}Shenzhen Institutes of Advanced Technology, Chinese Academy of Sciences, P.R. China~
Email: zhangyong@siat.ac.cn\\
\IEEEauthorrefmark{5}Department of Computer
Science, The University of Hong
Kong, P.R. China~
Email: yfzou@cs.hku.hk}}

\maketitle

\begin{abstract}
Network Function Virtualization (NFV) is a promising virtualization technology that has the potential to significantly reduce the expenses and improve the service agility. NFV makes it possible for Internet Service Providers (ISPs) to employ various Virtual Network Functions (VNFs) without installing new equipments. One of the most attractive approaches in NFV technology is a so-called Joint Placement and Allocation of Virtual Network Functions (JPA-VNF) which considers the balance between VNF investment with Quality of Services (QoS). We introduce a novel capability function to measure the potential of locating VNF instances for each server in the proposed OJPA-HS model. This model allows the servers in the network to be heterogeneous, at the same time combines and generalizes many classical JPA-VNF models. Despite its NP-hardness, we present a provable best-possible deterministic online algorithm based on dynamic programming (DP). To conquer the high complexity of DP, we propose two additional randomized heuristics, the Las Vegas (LV) and Monte Carlo (MC) randomized algorithms, which performs even as good as DP with much smaller complexity. Besides, MC is a promising heuristic in practice as it has the advantage to deal with big data environment. Extensive numerical experiments are constructed for the proposed algorithms in the paper.
\end{abstract}

\IEEEpeerreviewmaketitle

\section{Introduction}
Network Function Virtualization (NFV) is an emerging technology in which network functions are executed on generic-purpose
servers instead of proprietary software appliances. Such replacement makes it easier for Internet Service Providers (ISPs) to employ various Virtual Network Functions (VNFs) including firewalls, load balancers, network address translators, content filtering, deep packet inspection and so on. Also, with the advent of such evolution of networking, it is possible to employ the pre-mentioned VNFs without installing new equipments and thus more environmental friendly and cost efficient.

One of the most attractive and promising approaches with the potential to dramatically reduce the expenses is the Joint Placement and Allocation of Virtual Network Functions (JPA-VNF). This subject aims to find a good balance between VNF instances investment in the network in order to provide  specific service requirements with the Quality of Services (QoS), for example, packet loss probability. Generally, the objective is to optimally or nearly optimally allocate the VNF instances in the network in order to serve possible networking demands with certain QoS requirements.
Service requirements in NFV are characterized by network flows, and the packets need to pass through the ordered set of VNFs before reaching the destination. Such specific order is studied as the so-called Service Function Chain (SFC) in the NFV literature (see \cite{mgz16,wlw16,zxl17}). Vast works on NFV are done under the SFC constraints. For example, several such interesting problems are considered by Sallam et al. \cite{sg18}: 1) How to find an SFC-constrained shortest path between any pair of nodes? 2) What is the achievable SFC-constrained maximum flow? 3) How to place the VNFs such that the cost (the number of nodes to be virtualized) is minimized, while the maximum flow of the original
network can still be achieved even under the SFC constraint?

Various works are done from approximation algorithm perspective. Tomassilli et al. \cite{tgh18} study the problem of how to optimally place VNFs in the network such that the total deployment cost is minimized and at the same time all the SFC requirements of the flows are satisfied. Lukovszki et al. \cite{lrs18} studies approximation algorithms for the incremental deployment of a minimum number of middleboxes at optimal locations, such that capacity constraints at the middleboxes and length constraints on the communication routes are respected. And a NFV service distribution problem is also studied from approximation algorithm perspective by Feng et al. \cite{fl18}, whose goal is to determine the placement of VNFs, the routing of service flows, and the associated allocation of cloud and network resources that satisfy client demands with minimum cost. Particularly, they provide an $O(\varepsilon)$ approximation algorithm running in time $O(1/\varepsilon)$.

Heuristic algorithms are also popular in NFV study due to its practical efficiency. A near-optimal solution is solved by heuristic in \cite{lp16} for the real-time NFV, aiming to maximize the total number of requests assigned to the cloud for each SFC, while the deadline constraints are obeyed. Bari et al. \cite{bc15} provide a dynamic programming based heuristic to solve large instances of VNF placements. Trace driven simulations on real-world network topologies demonstrate the performance of the proposed heuristic is within 1.3 times of the optimal solution. Casado et al. \cite{ck10} consider the model with a single type of VNF and present a heuristic algorithm towards solving the placement problem.
And Gember et al. \cite{gk13} propose the design and implement for a novel architecture called Stratos, an orchestration layer for virtual
middleboxes, in order to overcome the lack of systematic tools to efficiently compose and provision in-the-cloud middleboxes.
Refer to \cite{ccw12} and \cite{hgj15} for more related work from different technical perspectives in the NFV literature.

\section{Related Work}\label{sec:related}
In general Joint Placement and Allocation of Virtual Network Functions (JPA-VNF), an undirected and connected network $G$ is given, for which we denote the set of nodes by $\cV$.  Each node in the network represents a private server or datacenter owned by a certain network operator or a service provider which has the potential to locate VNF instances. The service requirements are represented as a set of flow passing through the network, i.e., enter the network at a source node, get served along the flow path and leave the network afterward at the destination node. We denote the set of flow by $\cJ$, which is known before making decisions in previous work. Each flow $j\in \cJ$ has its specific amount of demand $d_j$ to serve, as well as the particular service path $P_j$, which is also known in JPA-VNF literature. Each $P_j$ is described as a particular sequence of nodes in the network. It is of huge difference for serving a requirement in a traditional network environment with that in a virtualized network environment. In the latter, we say a flow requirement $j$ is served or satisfied as long as the sum of VNF instances placed in the nodes which $P_j$ passing through are sufficient to serve $d_j$. While in the traditional environment, the network functions can only be placed in certain locations equipped on which are the proprietary hardware. Note distinct flows may overlap each other which may cause the crowd as well as the wasting of resources, see Fig. \ref{VNF} as an example. The goal in general JPA-VNF is mainly to balance the consumption of computing resources with the QoS.

\begin{figure}
\begin{center}
  \includegraphics[width=6cm,height=5cm]{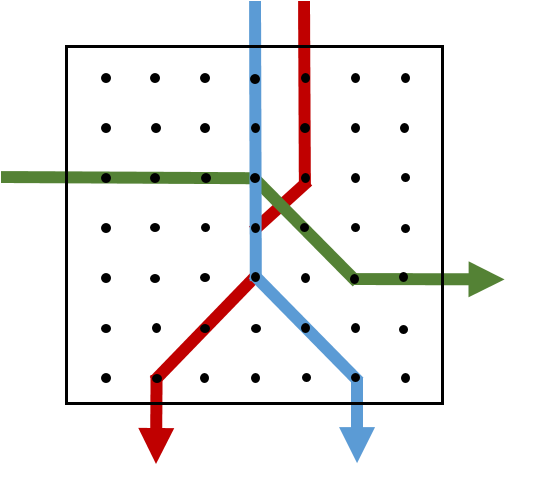}\\
  \caption{General JPA-VNF}\label{VNF}
\end{center}

\end{figure}

Many scenarios are proposed and considered for real-world requirements in JPA-VNF literature. For example in the work of Sang et al. \cite{sj2017}, they formulate a scenario of JPA-VNF as the following MILP problem:
\bea\label{info2017}
& \min\limits_{u_i,x_{ij}}  &\sum\limits_{i \in \cV} u_i \nn\\
& {\rm subject~to} &  \sum\limits_{i\in P_j} x_{ij} \ge d_j, ~~~~~~for~all~ j\in \cJ, \\
& & \sum\limits_{j\in \cJ} x_{ij}\le u_iR, ~~~~for~all~ i\in \cV, \nn\\
& & x_{ij}, u_i\in \mathbb{N}, ~~~~~~~~for~all~ i\in \cV,~ j\in \cJ.\nn
\eea

They take into account that an arbitrary VNF instance implemented on the virtual server has a limited fixed amount of computing resources, denoted by $R$ in program (\ref{info2017}). And multiple copies of $R$ can be added to the server by activating more VNF instances. In their model, one has to decide the placement and allocation of VNF instances such that all the flow demand can be served, at the same time the number of VNF instances implemented on the total network is minimized. In other words, the QoS is 100\% guaranteed and the wasting of computing resources is minimized in the model. They proposed an almost-tight approximation algorithm for this NP-hard problem.

However, a hidden assumption in the above work is that any virtual server in the given network has unlimited potential to activate VNF instances, while it is not always the truth. In the latter work of Sallam and Ji \cite{sj2019}, they remove this assumption and make a new one that once a virtual server is activated to be a VNF node, it has a fixed limited computing resources to serve flow demand. Each virtual server $i$ in the network has its specific amount $c_i$ of potential capacity and specific cost $b_i$ to be activated. In particular, they formulate the model as the following MILP:
\bea\label{info2019}
& \max\limits_{\cU,x_{ij}}  &\sum\limits_{j \in \cJ} d_j\cdot\mathbbm{1}_{\{\sum_{i\in\{P_j\cap\cU\}}x_{ij}\ge d_j\}} \nn\\
& {\rm subject~to} &  \sum\limits_{j\in \cJ} x_{ij} \le c_i\cdot\mathbbm{1}_{\{i\in\cU\}}, ~~~for~all~ i\in \cV, \\
& &  \sum\limits_{i\in \cU} b_i \le B,  \nn\\
& & x_{ij}\in \mathbb{N}, ~\cU\in\cV,~~~~~~~~~~for~all~ i\in \cV,~ j\in \cJ.\nn
\eea

The notation $\mathbbm{1}_{\{\mathcal{E}\}}$  is an indicator function for event $\mathcal{E}$ that equals to 1 if $\mathcal{E}$ happens and 0 otherwise. The capacity and budget constraints are considered in program (\ref{info2019}). In this model, one has to decide the subset $\cU$ of virtual servers in the network that will be activated to VNF nodes, as well as the allocation $x_{ij}$ of the service requirement $j\in\cJ$ along its specific flow path $P_j$, so as to maximize the total amount of fully served demand. In other words, this model is to maximize the QoS with limited budget of computing resources, and the crowd constraints are also considered in the model. A novel relaxation method is presented that modifies the objective function to a submodular one and through which a $(1-e^{-1})/2$-approximation algorithm is proposed in their work. Other interesting but less-related models and algorithms on JPA-VNF see \cite{cw18,cm14,cl15,mt14}.

\section{Our Models and Definitions}\label{sec:model}
We continue to use the notations defined in the last section if they cause no confusion. For the given undirected network $G$, let $\cV$ be the set of virtual servers that have the potential to locate VNF instances. For each distinct server $i\in\cV$, we formulate its capability of locating VNF instances by a specifically defined function $f_i(\cdot)$, the variable of which is the amount of demand assigned to server $i$, denoted by $u_i$. The value of $f_i(u_i)$ reflects a cost of using $u_i$ amount of computing resources of server $i$. By this function, we allow the virtual server to be heterogeneous and at the same time generalize a large family of capability functions, such as incremental function, concave function, step function, hard-capacitated function, etc. We only restrict $f_i(\cdot)$ to be a nondecreasing and left-continuous mapping from nonnegative reals to nonnegative reals with infinity. In particular, we make the following assumption.

\emph{Assumption 1}:  For each $i\in\cV$, $f_i(\cdot)$ satisfies the following:
\begin{enumerate}
  \item $f_i(\cdot)$ is a mapping from $\mathbb{R}^+$ to $\mathbb{R}^+\cup \{+\infty\}$;
  \item $f_i(x)\le f_i(y)$ for every $x\le y$;
  \item $\lim\limits_{x\rightarrow u^-}f_i(x)=f_i(u)$ for every $u$.
\end{enumerate}

We will show that our capability function is powerful enough to take both the crowd and the wasting of computing resources into consideration at the end of this section. Under \emph{Assumption 1}, we formulate the Joint Placement and Allocation of virtual network functions with Heterogeneous Servers (JPA-HS) as the following program:

\bea\label{model}
& \min\limits_{u_i,x_{ij}}  &\sum\limits_{i \in \cV} f_i(u_i) \nn\\
& {\rm subject~to} &  \sum\limits_{i\in P_j} x_{ij} \ge d_j, ~~~~~~for~all~ j\in \cJ, \\
& & \sum\limits_{j\in \cJ} x_{ij}\le u_i, ~~~~~~for~all~ i\in \cV, \nn\\
& & x_{ij}\in \mathbb{N}, ~~~~~~~~~~~~for~all~ i\in \cV,~ j\in \cJ.\nn
\eea

Through the following lemma, we show that by introducing such capability function $f_i(\cdot)$, our problem is well-defined because it is possible for one to make a globally best decision if the decision is not necessarily made in polynomial time.

\begin{lemma}
Under \emph{Assumption 1}, any feasible instance for program (\ref{model}) has a globally optimal solution.
\end{lemma}

\begin{proof}
Assume w.l.o.g. that the infimum $O^*$ of objective values for the set of solutions to program (\ref{model}) is a well-defined finite number, because otherwise any feasible solution is optimal (directly from the fact that no solution has a finite objective value). Then there must be a sequence of solutions $\{(u^k,x^k)\}^\infty_{k=1}$ such that $\lim_{k\rightarrow \infty}O(u^k,x^k)=O^*$, where $O(u,x)$ is the objective value of solution $(u,x)$. Note $u^k$ and $x^k$ are bounded from 1) of \emph{Assumption 1}. Since $O^*$ is finite, each $f_i(u_i)$ must be bounded for all $k> k_0$. Therefore an infinite subsequence $\{(u^{k_p},x^{k_p})\}^\infty_{p=1}$ can be picked such that for any $i\in\cV$, the following limits exist:
\begin{enumerate}
  \item $\lim_{p\rightarrow \infty}u^{k_p}=u^*$;
  \item $\lim_{p\rightarrow \infty}x^{k_p}=x^*$;
  \item $\lim_{p\rightarrow \infty}f_i(u_i^{k_p})=f_i^*$.
\end{enumerate}

Since $f_i$ is nondecreasing and left-continuous by assumption, $f_i(\lim_{p\rightarrow \infty}u_i^{k_p})\le f_i^*=\lim_{p\rightarrow \infty}f_i(u_i^{k_p})$, thus $(u^*,x^*)$ must be a globally optimal solution.
\end{proof}

However, if we restrict the decision be made in polynomial time, then the answer must be negative. We will briefly prove that JPA-HS is at least NP-hard by showing it generalizes program (\ref{info2017}).  In fact, our model of JPA-HS can also handle the capacity constraints in (\ref{info2019}).  Based on the NP-hardness of program (\ref{info2017}) proved in Sang et al. \cite{sj2017}, we are able to state the following.

\begin{lemma}
JPA-HS is at least NP-hard.
\end{lemma}
\begin{proof}
For all $i\in\cV$, by setting the following capability function, we are able to reduce JPA-VNF scenario (\ref{info2017}) to JPA-HS.
\begin{equation}
f_i(u_i)=
\left\{
\begin{array}{lr}
0, &~~~~~ u_i=0;\\
\lceil\frac{u_i}{R}\rceil, & u_i>0.
\end{array}
\right. \nn
\end{equation}

Easy to check that the above function satisfies \emph{Assumption 1}, particularly see Fig. \ref{step} as a description. Thus JPA-HS is harder than JPA-VNF which is proved to be NP-hard in \cite{sj2017}. Complete the proof.
\begin{figure}[h]
\begin{center}
  \includegraphics[width=6cm,height=5cm]{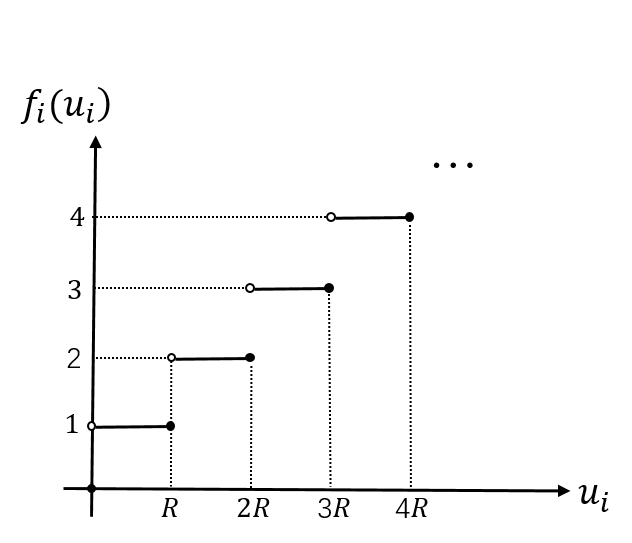}\\
  \caption{Capability function in JPA-VNF scenario (\ref{info2017})}\label{step}
\end{center}
\end{figure}

\end{proof}

Despite the above step function, we allow $f(\cdot)$ to locate in much larger function classes, for example linear function, concave function, hard-capacitated function etc. Taking a glance at the hard-capacitated one, note this is the capability function in JPA-VNF scenario (\ref{info2019}). For each $i\in\cV$, by setting
\begin{equation}
f_i(u_i)=
\left\{
\begin{array}{lr}
0, &~~~~~ u_i=0;\\
b_i, &~~~~~ 0<u_i\le c_i;\\
+\infty, & u_i>c_i,
\end{array}
\right. \nn
\end{equation}
we are able to reduce JPA-VNF scenario (\ref{info2019}) to JPA-HS neglecting the objective in that problem. Moreover, we allow distinct virtual servers have different capability functions in our model as long as it satisfies \emph{Assumption 1}.

For real-world application consideration, we mainly consider the online scenario of the described problem as it is always the case that we have continuously online service requirements. Only after the previous requirements get served, we start to consider the later one. No information for the sequence of service requirements are known before realized, i.e., we have no way to know the flow path $P_j$ as well $d_j$ before flow $j$ comes. And in NFV practical, decisions are costly and difficult to reverse. Therefore, we mainly consider the online version where no precious decision can be removed or changed. We give a name for this brand new model as Online Joint Placement and Allocation of virtual network functions with Heterogeneous Servers (OJPA-HS).

\section{Deterministic Online Algorithm}
In the last section, we already prove the NP-hardness of JPA-HS by reduction. However, it is far from the limit of the hardness for OJPA-HS. For the negative part, we construct an adversarial instance that no deterministic online algorithm can have a bounded competitive ratio. For the positive part, we propose a deterministic algorithm that from performance perspective always outputs the best-possible solution a deterministic online algorithm can do.

\subsection{Negative: an adversarial instance}
We reformulate the inputs of OJPA-HS by the defined recording matrix as following. Let $n:=|\cV|$ be the number of servers/nodes in the network and $m:=|\cJ|$ be the total number of flows at the end of the online instances. For notational convenience, we assume that the nodes are all pre-labeled in continuous order of positive integers and the flows are labeled in the same way according to the arriving time. Let matrix $A$ be the $n\times m$-dimensional matrix recording the nodes each arriving flow passing through. For each element $a_{ij}$ in $A$, it equals to 1 if flow $j$ passing through node $i$, and 0 otherwise. In this way, the matrix is revealed column by column over time. Consider a special case of OJPA-HS of which all the capability functions of virtual servers are hard-capacitated functions. In particular, we set up the instance by letting $d_j=1$ for all $j\in\cJ$, and
\begin{equation}
f_i(u_i)=
\left\{
\begin{array}{lr}
0, &~~~~~ u_i=0;\\
1, &~~~~~ 0<u_i\le 1;\\
+\infty, & u_i>1,
\end{array}
\right. \nn
\end{equation}
for all $i\in\cV$. Suppose that the $k$th flow passes through nodes $\{2k-1,2k\}$ for all $k\in\{1,\cdots,m/2\}$ and let $m=n$ be even for ease. See the recording matrix in Fig. \ref{matrix}, in which we ignore all zeros and highlight in red the allocations of VNF instances (i.e. the variable $x_{ij}$ who equals 1) output by an arbitrary deterministic online algorithm.
\begin{figure}[h]
\begin{center}
  \includegraphics[width=6cm,height=4.5cm]{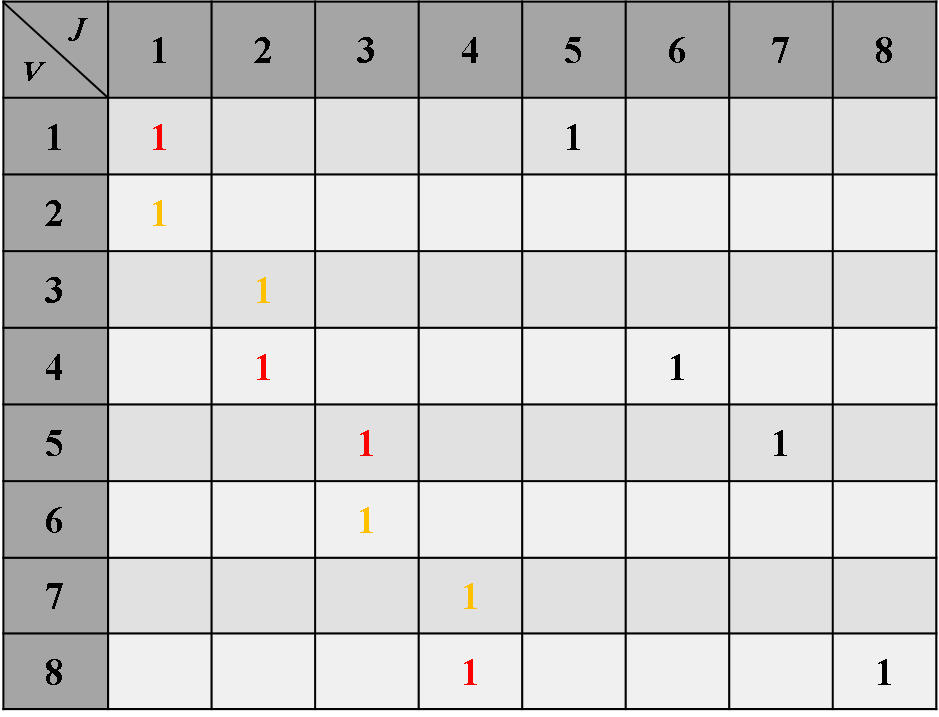}\\
  \caption{Recording matrix with $m=n=8$}\label{matrix}
\end{center}
\end{figure}

As long as the VNF allocations for the first half of flows (i.e. with label $\{1,\cdots,m/2\}$) are fixed/determined, then an adversarial set of the latter half of flows (i.e. with label $\{m/2+1,\cdots,m\}$) may pass through exactly the allocated nodes determined by the algorithm. Note this output can be any of the two,  $2k-1$ or $2k$ for the $k$-th flow in the above instance. Then even the best-possible decision one can make hits an unbounded cost of $+\infty$. From the other side, if we are informed by the information for all coming flows, we can make the best decision by switch our decision for all flows labeled $\{1,\cdots,m/2\}$. In the Fig. \ref{matrix} for example, we may switch our decision from red to orange and get a total cost of $m$. Thus the competitive ratio of any deterministic online algorithm in this example must be unbounded.
\begin{lemma}
No deterministic online algorithm for OJPA-HS has a bounded competitive ratio.
\end{lemma}

\subsection{Positive: a best-possible deterministic online algorithm}
Fortunately, we can make a best-possible online decision for OJPA-HS. By best-possible, we mean that we cannot behave better without changing our previous decisions in an arbitrary time slot. Assume that $\cJ$ arrives in sequence $\{j_1,j_2,\cdots,j_m\}$ and $\{v_1,v_2,\cdots,v_n\}$ is an arbitrary global label for $\cV$. Each server $v$ has a capability function $f_v(\cdot)$. The flow $j_k$ with an integral demand $d_{j_k}$, passes though a subset of $\cV$ denoted by  $P_{j_k}$. For notation convenience, we give a local label for the servers in $P_{j_k}$ as $\{v_1^k,v_2^k,\cdots,v_{n_k}^k\}$ such that it is a subsequence of $\{v_1,v_2,\cdots,v_n\}$, where $n_k=|P_{j_k}|$. We build a 3-dimensional dynamic programming to solve the best-possible allocation as following.

Considering the flow $j_k$, let $a^k(i,\omega)$ be the element of a 2-dimensional table $a^k$, where $i\in [1,n]$, $\omega\in [0,d_{j_k}]$ are both integers. The value of $a^k(i,\omega)$ represents the minimum cost of serving $\omega$ demand of $j_k$ after considering possible allocations on servers $\{v_1,v_2,\cdots,v_i\}$.  Of course we only consider the servers in  $P_{j_k}$ and for each $v_i\in \cV \setminus P_{j_k}$, we set $a^k(i,\omega)=+\infty$ or a sufficiently large number.
$\delta$ is an $n$-dimensional vector that updates over time. To specify its value, we record the value of $\delta$  by $\delta^k$ at the moment we make a decision for $j_k$. And the $i$-th element of $\delta^k$ represents the current allocation of $v_i$ after serving flow $j_k$. Let $\Delta f^k_i(\omega):=f_{v_i^k}(\omega+\delta_i^{k-1})-f_{v_i^k}(\delta_i^{k-1})$ be the cost increase after adding $\omega$ demand of flow $j_k$ to server $v_i^k$. And we assume w.l.o.g. that $\delta^0_i=0, f_i(0)=0$ for all $v_i\in \cV$. Compute for each $i$ such that $v^k_i\in P_{j_k}$ and each $\omega\in [0,d_{j_k}]$,
\begin{equation}
a^k(i,\omega)=
\left\{
\begin{array}{lr}
\Delta f^k_i(\omega),~~~~~~~~~~~~~~~~~~~~~~~~~~~~~~ i=1;\\~\\

\min\limits_{\delta^k_i\in [0,d_{j_k}]}\{a^k(i-1,\omega-\delta^k_i)+\Delta f^k_i(\delta^k_i)\},\\~~~~~~~~~~~~~~~~~~~~~~~~~~~~~~~~~~~~i\in[2,n_k].
\end{array}
\right. \nn
\end{equation}
 We can easily derive the recurrence and get the $n\times d_{j_k}$ table $a^k$. The value $a^k(n_k,d_{j_k})$ ($=a^k(n,d_{j_k})$) implies the minimum cost of serving $j_k$. And  each $\delta_i^k$ we choose in the recurrence for $i\in[2,n_k]$ constitutes the decision $\delta^k$ for the $k$-th flow.
Remember we have $m$ such tables. The final solution for OJPA-HS consists from each such $\delta^k$ and the objective value w.r.t. the algorithm is $\sum_{i\in[1,n]}f_{v_i}(\delta_i)$, where again, $\delta_i$ is the $i$-th element of $n$-dimensional vector $\delta$ updating throughout the algorithm.

For easy understanding, we present the pseudocode of the proposed dynamic programming as Algorithm \ref{DP}. Let $D:=\max_{k}d_{j_k}$ be the maximum demand of all the flows. We conclude the following theorem for the proposed deterministic online algorithm.

\begin{algorithm}[h]
\caption{Deterministic online algorithm}
\textbf{\label{DP}}
\nl \textit{Initialization:} $\delta^0=zeros(n)$;\vspace{-10pt} \\

\hrulefill\\
\nl \For {$k=1:m$}{
\nl \For {$i=1:n$}{
\nl \If{$v_i\notin P_{j_k}$}{
\nl \For {$\omega=1:d_{j_k}$}{
\nl $a^k(i,\omega)=+\infty$;\\
\nl $\delta^k_i=0$;\\}}
\nl \ElseIf {$v_i=v^k_1$}{

\nl \For {$\omega=1:d_{j_k}$}{
\nl $a^k(i,\omega)=\Delta f^k_i(\omega)$;\\}}

\nl \Else {
\nl \For {$\omega=1:d_{j_k}$}{
\nl $a^k(i,\omega)=\min\limits_{\delta^k_i}\{a^k(i-1,\omega-\delta^k_i)+\Delta f^k_i(\delta^k_i)\}$;\\
\nl $\delta^k(i)\leftarrow \delta^{k-1}(i)+\delta^k_i$;\\}}}}
\nl \Return Solution=$\delta^k$ for $j_k$\\
\nl~~~~~~~~ Value=$\sum_{i=1}^nf_i(\sum_{k=1}^m\delta^k(i))$.

\end{algorithm}

\begin{theorem}
We can in $O(mnD^2)$ time compute the best-possible allocation for OJPA-HS.
\end{theorem}
\begin{proof}
Recall that we compute a 2-dimensional table for $m$ flows. Each table $a^k$ has  $n\times d_{j_k}$ elements and thus is upper bounded by $nD$. All these end in $mnD$ times computation. For each element $a^k(i,\omega)$ we need to compute the minimum $\delta^k_i$ among $D$ options. Putting all together we conclude that the algorithm terminates in $O(mnD^2)$ time.

For optimality part, by best-possible we mean it is optimal at each step. That is, we make the best decision for each single arriving flow under the assumption that we cannot change the previous decisions. Suppose as before flows come in sequence $j_1\rightarrow j_2\rightarrow \cdots \rightarrow j_m$ and we focus on an arbitrary flow $j_k$. We build a $d_{j_k}$-tree for $j_k$, of which the root is a node represents for $j_k$, the depth is $|P_{j_k}|$ (the number of nodes $j_k$ passing through) and each node has $d_{j_k}+1$ branches. Each edge is equipped with a pair of weights $(a,b)$ if the edge connects a node with its $(a+1)$-th child, and $b$ is the cost increase to add $a$ demand of $j_k$ to $v_i^k$, i.e., $b=\Delta f^k_{v^k_i}(a)$. Fig \ref{tree} is an example of $d_{j_k}$-tree for $j_k$ where $d_{j_k}=2$ and $|P_{j_k}|=3$.

\begin{figure}[h]
  \includegraphics[width=8.6cm]{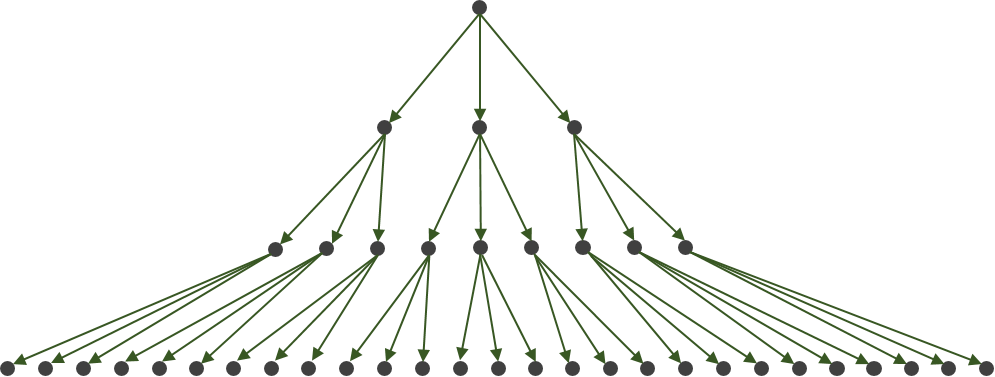}\\
  \caption{$d_{j_k}$-tree for $j_k$}\label{tree}
\end{figure}

Obviously, the $d_{j_k}$-tree contains all possible allocations of demand $d_{j_k}$. Algorithm \ref{DP} is a Breadth First Search (BFS) for $d_{j_k}$-tree to find out a path from the top to the bottom such that the sum of weight $b$ is minimized, at the same time the sum of weight $a$ is exactly $d_{j_k}$. The sum of weight $b$ along the chosen path implies the minimum cost increase of serving $j_k$, completing the proof.
\end{proof}

\section{Randomized Online Algorithm}
It is not the end of the story. In this section, we propose two types of randomized algorithms, both of which maintain a provably small time and space complexity compared to the previous deterministic algorithm. And from extensive numerical experiments we found that we do not lose too much by doing this. To explain the performance, we make a high-level evaluation by analyzing some interesting instances and exploring the reasons behind them.

\subsection{Las Vegas Randomized Algorithm}
Las Vegas Randomized Algorithm is such a randomized algorithm that always outputs a correct solution, in OJPA-HS which means, we always get all the flows fully served. We continue using the notations as in last section. Considering any online instance, we focus on the current arriving flow, w.o.l.g., $j_k$. We maintain a 2-dimensional array throughout the algorithm, according to which we execute our randomized algorithm, both the Las Vegas and Monte Carlo Randomized Algorithms. We denote the array by $P$, where $P$ is a $n\times d_{j_k}$ matrix. Let $P_{XY}$ be the $X$-th row $Y$-th column element of $P$. And this is the probability we pick the element. Particularly, we set $$P^{\prime}_{XY}=\frac{\mathbb{I}_{\{v_X\in P_{j_k}\}}\cdot Y}{f_X(\delta^{k-1}_X+Y)-f_X(\delta^{k-1}_X)},$$ where $\mathbb{I}_{\{v_X\in P_{j_k}\}}$ is the indicator function for the event $v_X\in P_{j_k}$, that is, it equals to 1 if $v_X\in P_{j_k}$ happens and 0 otherwise. To modify the array to be a probability distribution, we scale each $P^\prime_{XY}$ such that the sum of them are exactly one. That is,
$$P_{XY}=\frac{P^\prime_{XY}}{\sum_{X\in [1,n], Y\in[1,d_{j_k}]}P^\prime_{XY}}.$$

We design a Las Vegas Randomized Algorithm according to the constructed array $P$. At each round, we sample an element in $P$ according to the probability $P_{XY}$. Every time we pick an element of $P_{XY}$ means we assign $Y$ demand of $j_k$ to node $X$. For each sample, we need to update the entire array because the current allocation is changed, which may change the cost increase for adding another demand on the picked node. Repeat until the sum of picked value of $Y$ are sufficient to cover demand $d_{j_k}$. When the algorithm terminates, we obtain a feasible solution where all the flows are fully served. This is why we call the proposed algorithm  a Las Vegas Randomized Algorithm (LV). For easy understanding, we present the pseudocode of LV as Algorithm \ref{LV}.

\begin{algorithm}[h]
\caption{Las Vegas Randomized Algorithm}
\textbf{\label{LV}}
\nl \textit{Initialization:} $\delta^0=zeros(n)$;\vspace{-10pt} \\

\hrulefill\\
\nl \For {$k=1:m$}{
\nl $\delta^k\leftarrow \delta^{k-1}$;\\
\nl \Repeat {$\sum_{i\in[1,n]}(\delta^k(i)-\delta^{k-1}(i))\ge d_{j_k}$}{
\nl \For {$i=1:n$}{

\nl \For {$j=1:d_{j_k}$}{
\nl Compute $P^\prime_{ij}=\frac{\mathbb{I}_{\{v_i\in P_{j_k}\}}\cdot j}{f_i(\delta^k_i+j)-f_i(\delta^k_i)}$;\\
\nl Scale $P_{ij}=\frac{P^\prime_{ij}}{\sum_{i\in [1,n], j\in[1,d_{j_k}]}P^\prime_{ij}}$;
}
}
\nl Pick $(X,Y)$ according to the probability $P_{XY}$;\\
\nl $\delta^k(X)\leftarrow \delta^k(X)+Y$;\\
}
}
\nl \Return Solution=$\delta^k$-$\delta^{k-1}$ for $j_k$;\\
\nl~~~~~~~~ Value=$\sum_{i=1}^nf_i(\delta^k(i))$.
\end{algorithm}

In Algorithm \ref{LV}, we use $\delta^k$ to denote the current allocation for flows $\{j_1,j_2,\cdots,j_k\}$, which is different as in Algorithm \ref{DP}. Thus the allocation for $j_k$ is $\delta^k-\delta^{k-1}$. In fact, we can do slightly better than the above, in terms of improve the running time and the performance.
\begin{enumerate}
  \item Remember we update the entire array (line 7 and 8) after each sample, mainly because we should maintain the sum of one in order to keep $P$ a probability distribution. Actually we can make it a little different without change the sum. For example, every time we pick a $(X,Y)$ pair, we can only update the $X$-th row of $P$ and maintain the sum of that row.  That is, compute $P^\prime_{ij}$ for $i=X$ where X is picked in last sample, and then scale $P_{Xj}=\frac{P^\prime_{Xj}}{\sum_{j\in[1,d_{j_k}]}P^\prime_{Xj}}$. For $i\not=X$, we do not update $P_{ij}$.
  \item For the last $Y$ we picked for $j_k$, it may happen that when adding $Y$ demand, it exceed the demand $d_{j_k}$, which may cause a waste of the objective value. Towards that end, we remove the exceed part of demand at the last round. That is, we always keep $\sum_{i\in[1,n]}(\delta^k(i)-\delta^{k-1}(i))= d_{j_k}$ for $j_k$. In this way, we save (at least not waste) the output value, because of the nondecreasing property for each $f_i(\cdot)$.
\end{enumerate}

\subsection{Monte Carlo Randomized Algorithm}
Monte Carlo Randomized Algorithm is an algorithm that we sometimes lose part of the correctness but in turn sometimes we may get better solutions. Further, in our Monte Carlo algorithm, we can take a balance between them. Another advantage of Monte Carlo algorithm is that we can control the running time of it by setting parameters, in this way we may lose some correctness as mentioned. In Algorithm \ref{LV}, we do not care about the round we pick the random pair $(X,Y)$. While in our Monte Carlo Randomized Algorithm (MC), we restrict a fixed round for sampling the random pair. When the algorithm terminates, every flow has a probability not being fully served. Unless setting the number of rounds be $D$,  MC does not always outputs a feasible solution. We propose this algorithm only for practical consideration in the case that the decision maker is allowed to have a certain percentage of fail.

The main body of the MC is quite similar with LV. That is, we maintain the same size array $P$ throughout the algorithm, compute and update $P$ in the same way as we do in LV. The only difference is the break condition for the algorithm. For flow $j_k$, we replace the break condition $\sum_{i\in[1,n]}(\delta^k(i)-\delta^{k-1}(i))\ge d_{j_k}$ with the number of rounds $r$. That is, we do not care if the flow is fully served or not, whenever it execute $r$ rounds of sampling, we break the loop for the flow. When break, we give a label for each flow as SUCCESS or FAIL, depending on the flow is fully served or not. When the algorithm terminates, there is a new output that will report the FAIL rate of serving the flows, defined as $\frac{\#{\rm FAIL}}{m}$. This is the general idea for the proposed MC. There are still some possible modifications that may improve the performance of MC.
\begin{enumerate}
  \item Remember what we do for LV is to remove the exceed demand at the last pick for each flow. We keep doing this for MC, but only for the SUCCESS flows. For the flow labeled FAIL, we remove all the allocation of it and of course no exceed demand can be removed.
  \item The setting of parameter $r$. In the above description, we set a uniform $r$ for each flow, which is not working perfectly. Thinking about the two flow $j$ and $j^\prime$ with demand $d_j>d_{j^\prime}$. For a fixed number of sampling, it is more likely that $j$ will not get fully served than $f^\prime$. From the Morkov inequality we know the probability of a random variable exceeding a fixed constant is proportional with its expectation, thus we can do better by setting the rounds of sampling for $j_k$ proportional with $d_{j_k}$. Particularly, we set the rounds for $j_k$ be $\lceil \frac{d_{j_k}}{r}\rceil$ and here $r$ is a parameter to be determined.
\end{enumerate}

We present the pseudocode of the MC as Algorithm \ref{MC}.

\begin{algorithm}[h]
\caption{Monte Carlo Randomized Algorithm}
\textbf{\label{MC}}
\nl \textit{Initialization:} $\delta^0=zeros(n)$;\vspace{-5pt} \\
\hrulefill\\
\nl \For {$k=1:m$}{
\nl $\delta^k\leftarrow \delta^{k-1}$;\\
\nl \For {round=$1:\lceil\frac{d_{j_k}}{r}\rceil$}{
\nl \For {$i=1:n$}{
\nl \For {$j=1:d_{j_k}$}{
\nl Compute $P^\prime_{ij}=\frac{\mathbb{I}_{\{v_i\in P_{j_k}\}}\cdot j}{f_i(\delta^k_i+j)-f_i(\delta^k_i)}$;\\
\nl Scale $P_{ij}=\frac{P^\prime_{ij}}{\sum_{i\in [1,n], j\in[1,d_{j_k}]}P^\prime_{ij}}$;
}
}
\nl Pick $(X,Y)$ according to the probability $P_{XY}$;\\
\nl $\delta^k(X)\leftarrow \delta^k(X)+Y$;\\
}
\nl \If {$\sum_{i\in[1,n]}(\delta^k(i)-\delta^{k-1}(i))\ge d_{j_k}$}{
\nl Label $j_k$ SUCCESS;}
\nl \Else{
\nl Label $j_k$ FAIL;
}
}
\nl \Return Solution=$\delta^k$-$\delta^{k-1}$ for $j_k$;\\
\nl~~~~~~~~ Value=$\sum_{i=1}^nf_i(\delta^k(i))$;\\
\nl~~~~~~~~ FAIL rate=$\frac{\#(FAIL)}{m}$.

\end{algorithm}

\subsection{Analysis}
\subsubsection{Space-Complexity}
The only space we maintain throughout the algorithms (both LV and MC) are for the 2-dimensional array $P$. We construct a $n \times d_{j_k}$ matrix for $j_k$. When $j_{k+1}$ arrives, we need only $n \times \max\{d_{j_k},d_{j_{k+1}}\}$ space, because we cover and replace the elements for $j_k$ with that for $j_{k+1}$ as the online algorithm proceeds.  This ends up in a space complexity of $O(nD)$ for the proposed LV and MC, where $D=\max_{p\in [1,m]}d_{j_p}$.

\subsubsection{Time-Complexity}
For the proposed randomized algorithms LV and MC, the running time spent are dominated by the computation of array $P$. The value $P^\prime_{XY}$ can be calculated rapidly when the variable $X$ and $Y$ are ready, and for $P_{XY}$ we divide every element of $P^\prime_{XY}$ by the sum of all $P^\prime_{XY}$ in the table. Thus we have constant time for the computation of each $P_{XY}$.  Remember we have at most $nD$ elements for one flow, which gives a total time complexity of $O(mnD)$.

\subsubsection{Performance Evaluation}
From the deterministic online algorithm, we notice that to choose the best-possible VNF allocations for each flow may in turn end up with a bad solution, even when the capability functions are not hard-capacitated ones. Suppose we have two types of servers: one P-Server and $m$ Q-Servers. Each flow passes through the only P-Server and one of the $m$ Q-Servers, the demand of which are all ones. Let $f_P(u_P)=1+\varepsilon\cdot u_P$ and $f_Q(u_Q)=u_Q$. Employing Algorithm \ref{DP} to this instance we get a solution that allocates the demand of $j_k$ on the Q-Server that $j_k$ passes through, which end up in a total cost of $m$. For this instance on the other hand, we can easily get an optimal cost of $1+\varepsilon\cdot m$ by allocating all the demand to the P-Server. This gives a performance guarantee of $+\infty$ as $m$ grows to infinitely large.

\begin{figure}[h]
\begin{center}
  \includegraphics[width=5cm,height=4cm]{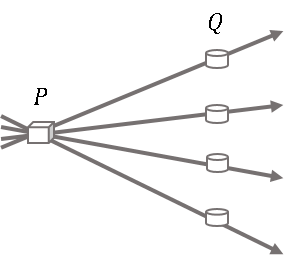}\\
  \caption{P-Server and Q-Servers}\label{PQ}
\end{center}
\end{figure}

Meanwhile in our randomized algorithms, we do not always pursue the best-possible decision for each flow. As we see in the sampling step, we choose the server $X$ as well as the correlative served demand $Y$ according to the probability of $P_{XY}$. Ignoring the indicator $\mathbb{I}_{\{v_X\in P_{j_k}\}}$ that force the probability to be zero for the servers not in $P_{j_k}$, the remaining part of $P^\prime_{XY}$ is the reciprocal of the marginal cost of serving $Y$ in server $X$, i.e., $\frac{{f_X(\delta^{k-1}_X+Y)-f_X(\delta^{k-1}_X)}}{Y}$. With such settings, the algorithms are more likely to choose the $(X,Y)$ pair with small marginal cost, but not always the smallest one. In the numerical experiments we show that by doing this, it performs even as good as the deterministic online algorithm which costs $O(mnD^2)$ time, as the randomized one only costs $O(mnD)$ time.

\begin{theorem}
The proposed randomized heuristic algorithms terminate in time $O(mnD)$ with space complexity $O(nD)$.
\end{theorem}

\section{Numerical Results}
In this section, we evaluate the performances of our proposed Algorithm \ref{DP} (DP), Algorithm  \ref{LV} (LV) and Algorithm \ref{MC} (MC). Our environment for experiments is Intel(R) Xeon(R) CPU E5-2620 v4 @ 2.10GHz with 64GB memory. We construct extensive numerical experiments to analyze different impacts of the proposed algorithms as well as the parameter settings. Mainly, we make cross comparisons between DP, LV and MC. Note MC is different from DP and LV because we are not supposed to get feasible solutions. Thus it makes no sense to compare the output value of MC with the other two algorithms. For MC, we focus on the running time comparison between it with LV (and thus with DP because we also have the running time comparison for LV with DP).At last, by setting different parameter $r$, we are able to see the impact of the rounds of sampling in MC.
Note in the following experiments, the capability functions of the servers are generated among linear function, concave function, hard capacitated function and step function with equal probabilities. Note in some of the result figures, we smooth the lines for presentation considerations.

 \textbf{1) DP vs. LV for 35 random instances}
In this experiment, we randomly generate 35 instances with given support. Particularly, we set the size of the network to be 50 nodes,
passing through which are 400 flows with maximum demand 200. The paths of the flows are also generated uniformly at random within the network. We run DP and LV once for each generated instance. As was expected, the LV runs much faster than DP for the same instance, but we found that the output value of them are surprisingly close, see Fig. \ref{1}.

%
\begin{figure}[h]
\begin{center}
\begin{tikzpicture}[scale=0.7]
  \begin{axis}[ 
    xlabel={Instance number},
    ylabel={Running time (s)},
    xmin=1, xmax=35,
    ymin=0, ymax=160,
    ytick={0,10,20,30,40,50,60,70,80,90,100,110,120,130,140,150},
    legend columns=3,
  ]

\addplot[color=blue,
                mark options={solid},
                smooth]  	
coordinates{ 
(	1	,	6.901508093	)
(	2	,	4.813677073	)
(	3	,	5.741599798	)
(	4	,	7.115626097	)
(	5	,	7.712809086	)
(	6	,	5.968748093	)
(	7	,	5.212610006	)
(	8	,	5.500578165	)
(	9	,	6.208392143	)
(	10	,	6.240756035	)
(	11	,	7.431531191	)
(	12	,	7.112971067	)
(	13	,	6.933593988	)
(	14	,	6.347028017	)
(	15	,	6.521520138	)
(	16	,	5.583968163	)
(	17	,	6.746335983	)
(	18	,	6.607235909	)
(	19	,	5.230648041	)
(	20	,	6.157793045	)
(	21	,	6.192298889	)
(	22	,	6.905430794	)
(	23	,	6.308185101	)
(	24	,	5.519074917	)
(	25	,	6.292862892	)
(	26	,	5.700720787	)
(	27	,	5.05667305	)
(	28	,	6.192816019	)
(	29	,	8.145206928	)
(	30	,	6.581865788	)
(	31	,	5.724907875	)
(	32	,	5.713963985	)
(	33	,	5.75338006	)
(	34	,	5.861660004	)
(	35	,	6.799412012	)};
\addlegendentry{LV} 
\addplot[color=red,
                mark options={solid},
                smooth]  	
coordinates{ 
(	1	,	134.9827189	)
(	2	,	91.9416132	)
(	3	,	102.933243	)
(	4	,	134.5345271	)
(	5	,	124.8321061	)
(	6	,	98.85265398	)
(	7	,	95.8620379	)
(	8	,	110.8394539	)
(	9	,	93.67440295	)
(	10	,	106.724118	)
(	11	,	142.9506669	)
(	12	,	132.0647731	)
(	13	,	124.1769969	)
(	14	,	119.7768691	)
(	15	,	125.0897682	)
(	16	,	113.65839	)
(	17	,	131.9591119	)
(	18	,	99.59849405	)
(	19	,	98.94624519	)
(	20	,	117.575985	)
(	21	,	119.9141059	)
(	22	,	139.8617671	)
(	23	,	125.1433349	)
(	24	,	108.090009	)
(	25	,	117.1657431	)
(	26	,	113.5541859	)
(	27	,	90.94951797	)
(	28	,	120.0865619	)
(	29	,	126.2797122	)
(	30	,	132.793746	)
(	31	,	111.0662482	)
(	32	,	103.2034149	)
(	33	,	99.35886502	)
(	34	,	110.2726331	)
(	35	,	134.957993	)};
\addlegendentry{DP}

  \end{axis}
\end{tikzpicture}

\begin{tikzpicture}[scale=0.7]
  \begin{axis}[ 
    xlabel={Instance number},
    ylabel={Cost of the solution},
    xmin=1, xmax=35,
    ymin=0, ymax=6000,
    legend columns=3,
  ]

\addplot[color=blue,
                mark options={solid},
                smooth]  	
coordinates{ 
(	1	,	2286.858011	)
(	2	,	3073.69825	)
(	3	,	2535.387696	)
(	4	,	1241.251067	)
(	5	,	801.8604224	)
(	6	,	2802.767142	)
(	7	,	2054.813368	)
(	8	,	4325.155842	)
(	9	,	3195.666809	)
(	10	,	2592.672771	)
(	11	,	2656.893595	)
(	12	,	2606.414984	)
(	13	,	1665.950688	)
(	14	,	3138.438343	)
(	15	,	2099.597211	)
(	16	,	2422.229866	)
(	17	,	1976.097412	)
(	18	,	1744.328846	)
(	19	,	1944.722739	)
(	20	,	2726.75238	)
(	21	,	1022.980739	)
(	22	,	1445.85999	)
(	23	,	1782.914953	)
(	24	,	2495.117651	)
(	25	,	1918.474781	)
(	26	,	1904.785109	)
(	27	,	5254.659929	)
(	28	,	1597.002596	)
(	29	,	946.372548	)
(	30	,	812.3697595	)
(	31	,	1547.253012	)
(	32	,	1515.830364	)
(	33	,	3531.610686	)
(	34	,	3357.875256	)
(	35	,	3138.210536	)};
\addlegendentry{ LV } 
\addplot[color=red,
                mark options={solid},
                smooth]  	
coordinates{ 
(	1	,	1757.282847	)
(	2	,	3125.390806	)
(	3	,	2350.077361	)
(	4	,	1179.574947	)
(	5	,	1433.507912	)
(	6	,	2124.772063	)
(	7	,	2082.664002	)
(	8	,	3727.162796	)
(	9	,	2037.243177	)
(	10	,	2159.237625	)
(	11	,	3172.294232	)
(	12	,	2412.350397	)
(	13	,	2063.958737	)
(	14	,	2267.550966	)
(	15	,	1919.744641	)
(	16	,	1927.007074	)
(	17	,	1244.040438	)
(	18	,	1558.253527	)
(	19	,	1792.264432	)
(	20	,	2173.438057	)
(	21	,	1141.744178	)
(	22	,	1649.318766	)
(	23	,	2059.649135	)
(	24	,	2644.528674	)
(	25	,	1681.853395	)
(	26	,	2230.26741	)
(	27	,	5080.140935	)
(	28	,	2105.170867	)
(	29	,	1167.745489	)
(	30	,	1128.877857	)
(	31	,	1481.861519	)
(	32	,	1467.880358	)
(	33	,	2185.393321	)
(	34	,	2323.920455	)
(	35	,	3138.012507	)};
\addlegendentry{ DP }

  \end{axis}
\end{tikzpicture}

\end{center}

  \caption{Experiment 1}\label{1}
\end{figure}
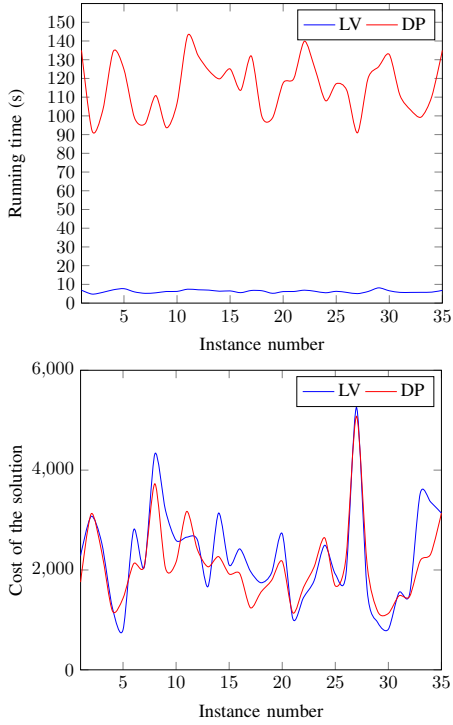

\textbf{2) DP vs. LV expected values (of 30 runs) for 30 random instances}
In case that the result 1) happens accidentally, we construct experiment 2). For the same setting, we generate another 30 instances, and for each instance, we run DP once and LV 30 times. In this way, we compare the expected value for LV of 30 runs with the output value of DP. Besides, we make a ratio between the cost of LV with that of DP to see how close in each instance they are. See Fig. \ref{2}, the recorded maximum ratio of the performance is 2.3057 in instance 13, and the minimum is 0.8858 in instance 24, which means LV performs even better than DP in this instance.

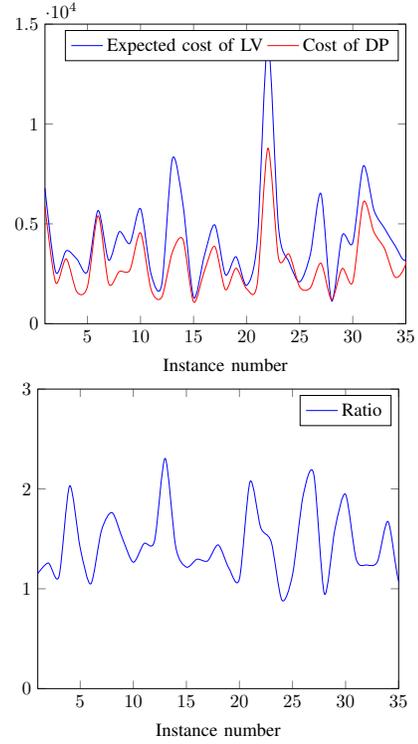
\begin{figure}[h]
\begin{center}
\begin{tikzpicture}[scale=0.7]
  \begin{axis}[ 
    xlabel={Instance number},
    xmin=1, xmax=35,
    ymin=0, ymax=15000,
    legend columns=3,
  ]

\addplot[color=blue,
                mark options={solid},
                smooth]  	
coordinates{ 
(	1	,	6777.14578	)
(	2	,	2603.758384	)
(	3	,	3651.217161	)
(	4	,	3238.289698	)
(	5	,	2611.179469	)
(	6	,	5670.353036	)
(	7	,	3184.694363	)
(	8	,	4604.783059	)
(	9	,	4023.950621	)
(	10	,	5757.807445	)
(	11	,	2474.712357	)
(	12	,	1961.338533	)
(	13	,	8249.202022	)
(	14	,	5973.869589	)
(	15	,	1329.336714	)
(	16	,	3371.155367	)
(	17	,	4947.645434	)
(	18	,	2454.185092	)
(	19	,	3352.958356	)
(	20	,	1920.709856	)
(	21	,	4146.72481	)
(	22	,	14181.71963	)
(	23	,	4843.163834	)
(	24	,	3081.077933	)
(	25	,	2097.001326	)
(	26	,	3470.811257	)
(	27	,	6510.457126	)
(	28	,	1149.784092	)
(	29	,	4421.606231	)
(	30	,	4061.935532	)
(	31	,	7900.436162	)
(	32	,	5677.953368	)
(	33	,	4768.538543	)
(	34	,	3897.234702	)
(	35	,	3237.292674	)
(	36	,	5813.043304	)};
\addlegendentry{ Expected cost of LV } 
\addplot[color=red,
                mark options={solid},
                smooth]  	
coordinates{ 
(	1	,	5881.663485	)
(	2	,	2069.800208	)
(	3	,	3248.21601	)
(	4	,	1595.327266	)
(	5	,	1855.833015	)
(	6	,	5402.069349	)
(	7	,	2010.880288	)
(	8	,	2612.385486	)
(	9	,	2687.58251	)
(	10	,	4549.483136	)
(	11	,	1703.701497	)
(	12	,	1325.080052	)
(	13	,	3577.764511	)
(	14	,	4214.955723	)
(	15	,	1092.321476	)
(	16	,	2601.279385	)
(	17	,	3868.025423	)
(	18	,	1705.339275	)
(	19	,	2777.932773	)
(	20	,	1746.046198	)
(	21	,	2000.833418	)
(	22	,	8787.61937	)
(	23	,	3301.906499	)
(	24	,	3478.455521	)
(	25	,	1838.288337	)
(	26	,	1795.587287	)
(	27	,	3030.097745	)
(	28	,	1209.214523	)
(	29	,	2765.114041	)
(	30	,	2087.927903	)
(	31	,	6092.772432	)
(	32	,	4584.676625	)
(	33	,	3748.189892	)
(	34	,	2328.579467	)
(	35	,	2998.86413	)
(	36	,	5464.365879	)};
\addlegendentry{ Cost of DP }

  \end{axis}
\end{tikzpicture}

\begin{tikzpicture}[scale=0.7]
  \begin{axis}[ 
    xlabel={Instance number},
    xmin=1, xmax=35,
    ymin=0, ymax=3,
    legend columns=3,
  ]

\addplot[color=blue,
                mark options={solid},
                smooth]  	
coordinates{ 
(	1	,	1.152249835	)
(	2	,	1.257975709	)
(	3	,	1.124068458	)
(	4	,	2.029859181	)
(	5	,	1.407012079	)
(	6	,	1.049663133	)
(	7	,	1.583731454	)
(	8	,	1.76267365	)
(	9	,	1.497237984	)
(	10	,	1.265595953	)
(	11	,	1.452550439	)
(	12	,	1.480166071	)
(	13	,	2.305686133	)
(	14	,	1.417303047	)
(	15	,	1.216983043	)
(	16	,	1.295960513	)
(	17	,	1.27911399	)
(	18	,	1.439118378	)
(	19	,	1.20699766	)
(	20	,	1.100033812	)
(	21	,	2.072498777	)
(	22	,	1.613829529	)
(	23	,	1.466778007	)
(	24	,	0.885760336	)
(	25	,	1.140735805	)
(	26	,	1.932967159	)
(	27	,	2.148596406	)
(	28	,	0.950852037	)
(	29	,	1.599068308	)
(	30	,	1.945438597	)
(	31	,	1.296689849	)
(	32	,	1.238463218	)
(	33	,	1.272224375	)
(	34	,	1.673653297	)
(	35	,	1.079506284	)
(	36	,	1.063809312	)};
\addlegendentry{ Ratio }

  \end{axis}
\end{tikzpicture}

\end{center}

  \caption{Experiment 2}\label{2}
\end{figure}


 \textbf{3) LV vs. MC w.r.t. running time}
As mentioned, it makes no sense to compare the output value of LV with that of MC. We construct this experiment mainly to see the running time performance of LV and MC for variant sizes of the inputs. As time-complexity analysis for the proposed algorithms says, the randomized algorithm performs much better than the deterministic one in terms of the running time. Thus in this experiment, we allow the inputs to be much larger than that in the previous two experiments. Particularly, we build three different types of instances:

\textbf{A-type}: 50 nodes, 100 flows with maximum demand 200;

\textbf{B-type}: 50 nodes, 500 flows with maximum demand 200;

\textbf{C-type}: 50 nodes, 500 flows with maximum demand 1000.

For each type, we randomly generate 30 instances and for each instance, run LV and MC once respectively. See Fig. \ref{3}, for the A-type instances, the running time of the LV and MC are very close. However, with the instances grow larger and larger, the running time of them starts to separate. And for C-type instances, MC outperforms LV in terms of running time.

\begin{figure}[h]
\begin{center}
\begin{tikzpicture}[scale=0.7]
  \begin{axis}[ 
    xlabel={A-type instance number},
    ylabel={Running time (s)},
    xmin=1, xmax=30,
    ymin=0,
    legend columns=3,
  ]

\addplot[color=blue,
                mark options={solid},
                smooth]  	
coordinates{ 
(	1	,	2.24683094	)
(	2	,	1.948189974	)
(	3	,	1.917916059	)
(	4	,	2.566096067	)
(	5	,	1.798512936	)
(	6	,	2.465445995	)
(	7	,	2.120037079	)
(	8	,	2.429563046	)
(	9	,	2.389698982	)
(	10	,	2.259896994	)
(	11	,	2.225565195	)
(	12	,	2.1711061	)
(	13	,	2.244111061	)
(	14	,	1.962957144	)
(	15	,	2.376837015	)
(	16	,	2.023880959	)
(	17	,	2.31703186	)
(	18	,	2.487796068	)
(	19	,	2.171406031	)
(	20	,	2.959280014	)
(	21	,	1.975306034	)
(	22	,	1.825998068	)
(	23	,	2.249426842	)
(	24	,	2.710040092	)
(	25	,	2.060498953	)
(	26	,	1.595869064	)
(	27	,	2.349334955	)
(	28	,	2.158725023	)
(	29	,	2.110038042	)
(	30	,	2.085901022	)};
\addlegendentry{ LV } 
\addplot[color=red,
                mark options={solid},
                smooth]  	
coordinates{ 
(	1	,	2.150655031	)
(	2	,	1.780858994	)
(	3	,	1.97309804	)
(	4	,	1.95549798	)
(	5	,	1.825937033	)
(	6	,	2.030739784	)
(	7	,	2.016355991	)
(	8	,	2.08344698	)
(	9	,	2.071329832	)
(	10	,	1.823652983	)
(	11	,	1.921290874	)
(	12	,	2.071250916	)
(	13	,	1.988641024	)
(	14	,	1.882141113	)
(	15	,	1.808807135	)
(	16	,	1.924365044	)
(	17	,	1.93894887	)
(	18	,	2.032951117	)
(	19	,	1.717208862	)
(	20	,	2.172821999	)
(	21	,	1.97200799	)
(	22	,	1.493100882	)
(	23	,	1.919817925	)
(	24	,	2.01733613	)
(	25	,	1.873551846	)
(	26	,	1.646770954	)
(	27	,	1.915464878	)
(	28	,	1.991046906	)
(	29	,	2.180083036	)
(	30	,	1.618979931	)};
\addlegendentry{ MC }

  \end{axis}
\end{tikzpicture}

\begin{tikzpicture}[scale=0.7]
  \begin{axis}[ 
    xlabel={B-type instance number},
    ylabel={Running time (s)},
    xmin=1, xmax=30,
    ymin=0, 
    legend columns=3,
  ]

\addplot[color=blue,
                mark options={solid},
                smooth]  	
coordinates{ 
(	1	,	12.25847507	)
(	2	,	12.86228299	)
(	3	,	13.39141202	)
(	4	,	13.91389704	)
(	5	,	12.31361604	)
(	6	,	13.69489598	)
(	7	,	13.24568486	)
(	8	,	14.53527403	)
(	9	,	13.38440108	)
(	10	,	12.94687486	)
(	11	,	14.01335001	)
(	12	,	13.61895514	)
(	13	,	13.82609391	)
(	14	,	14.93189096	)
(	15	,	13.77835989	)
(	16	,	11.80423498	)
(	17	,	14.25587201	)
(	18	,	12.00384307	)
(	19	,	12.79927015	)
(	20	,	14.22226882	)
(	21	,	13.0552671	)
(	22	,	12.74454999	)
(	23	,	14.53364897	)
(	24	,	11.43644905	)
(	25	,	12.16950202	)
(	26	,	13.41387415	)
(	27	,	13.52536511	)
(	28	,	11.62638521	)
(	29	,	14.78723502	)
(	30	,	13.95087194	)};
\addlegendentry{ LV } 

\addplot[color=red,
                mark options={solid},
                smooth]  	
coordinates{ 
(	1	,	10.27765799	)
(	2	,	10.11974382	)
(	3	,	12.64932013	)
(	4	,	11.95867991	)
(	5	,	10.06269503	)
(	6	,	12.39790201	)
(	7	,	11.98604107	)
(	8	,	11.27508092	)
(	9	,	12.03588915	)
(	10	,	10.44616199	)
(	11	,	11.63688397	)
(	12	,	12.44334388	)
(	13	,	12.28906989	)
(	14	,	11.72954416	)
(	15	,	12.74911618	)
(	16	,	10.54983592	)
(	17	,	13.30214405	)
(	18	,	10.26099205	)
(	19	,	11.76380897	)
(	20	,	12.25104785	)
(	21	,	11.33705902	)
(	22	,	10.8460598	)
(	23	,	12.00582409	)
(	24	,	9.726355076	)
(	25	,	10.93572903	)
(	26	,	11.65464306	)
(	27	,	11.59512401	)
(	28	,	9.821372986	)
(	29	,	11.73369884	)
(	30	,	11.4103539	)};
\addlegendentry{ MC } 

  \end{axis}
\end{tikzpicture}

\begin{tikzpicture}[scale=0.7]
  \begin{axis}[ 
    xlabel={C-type instance number},
    ylabel={Running time (s)},
    xmin=1, xmax=30,
    ymin=0, 
    legend columns=3,
  ]

\addplot[color=blue,
                mark options={solid},
                smooth]  	
coordinates{ 
(	1	,	175.7262502	)
(	2	,	202.2482071	)
(	3	,	183.1258421	)
(	4	,	196.2584569	)
(	5	,	197.5780962	)
(	6	,	152.878696	)
(	7	,	218.34763	)
(	8	,	169.007982	)
(	9	,	146.9295151	)
(	10	,	192.3621461	)
(	11	,	153.5322959	)
(	12	,	114.5563841	)
(	13	,	165.2379441	)
(	14	,	173.977824	)
(	15	,	182.460006	)
(	16	,	164.963088	)
(	17	,	150.7730699	)
(	18	,	208.9440632	)
(	19	,	128.3400991	)
(	20	,	214.096565	)
(	21	,	250.775264	)
(	22	,	162.3928468	)
(	23	,	170.02074	)
(	24	,	158.2389922	)
(	25	,	190.4540238	)
(	26	,	189.0502369	)
(	27	,	157.0763788	)
(	28	,	185.7077892	)
(	29	,	160.808255	)
(	30	,	193.8924742	)};
\addlegendentry{ LV } 

\addplot[color=red,
                mark options={solid},
                smooth]  	
coordinates{ 
(	1	,	80.41347098	)
(	2	,	94.51186705	)
(	3	,	88.26558805	)
(	4	,	101.8878391	)
(	5	,	85.83175015	)
(	6	,	78.92328215	)
(	7	,	83.21003294	)
(	8	,	84.43478513	)
(	9	,	76.80827403	)
(	10	,	85.75872993	)
(	11	,	83.45407701	)
(	12	,	72.76563001	)
(	13	,	92.80385995	)
(	14	,	86.55959082	)
(	15	,	80.61918402	)
(	16	,	87.78751898	)
(	17	,	88.01257396	)
(	18	,	88.34128308	)
(	19	,	89.63877487	)
(	20	,	91.78777313	)
(	21	,	92.81946301	)
(	22	,	84.10102391	)
(	23	,	78.570889	)
(	24	,	84.73576188	)
(	25	,	88.68194699	)
(	26	,	85.98008895	)
(	27	,	88.55339408	)
(	28	,	91.74237108	)
(	29	,	85.11517096	)
(	30	,	85.06592512	)};
\addlegendentry{ MC }

  \end{axis}
\end{tikzpicture}

\end{center}

  \caption{Experiment 3}\label{3}
\end{figure}
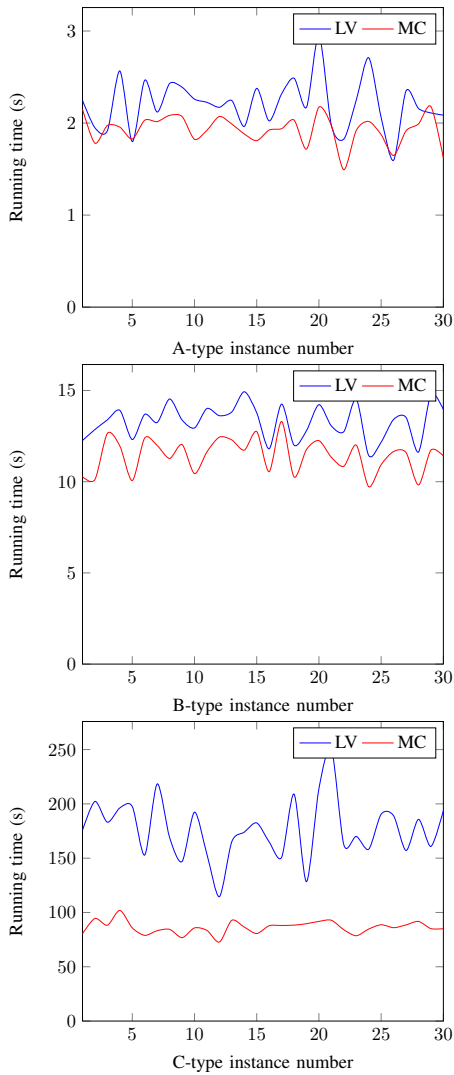

%
%

\textbf{4) $r$ vs. FAIL rate in MC}
In subsection B of section V, we roughly discuss the impact of the rounds of sampling in MC. And in order to have better performance, we set the number of rounds be proportional with the flow demand as $\lceil \frac{d_j}{r}\rceil$. In this experiment, we are able to see the relation between FAIL rate with the setting of $r$. Note the larger $r$ is, the smaller rounds of sampling we have, and thus the higher FAIL rate may occur. Particularly, the instances are randomly generated under settings 50 nodes, 500 flows with maximum demand 1000. And $r$ grows from 5 to 50 with step length 5. Results see Fig. \ref{4}.

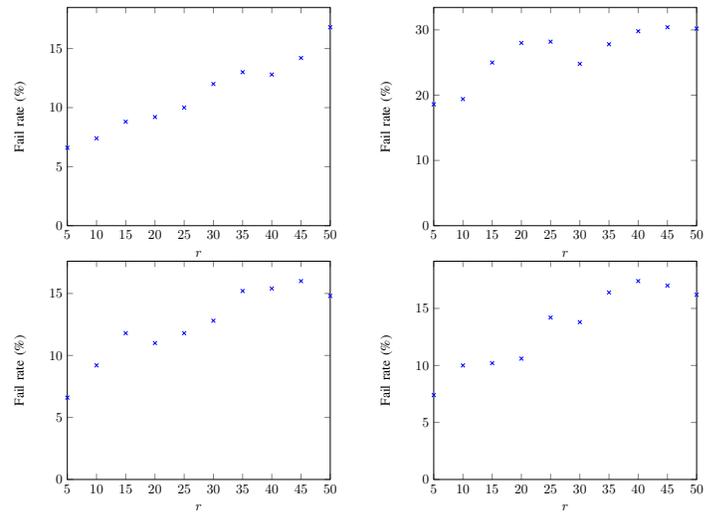
\begin{figure}[h]

\begin{minipage}{0.45\linewidth}

\begin{tikzpicture}[scale=0.51]
  \begin{axis}[ 
    xlabel={$r$},
    ylabel={Fail rate (\%)},
    xmin=5, xmax=50,
    ymin=0,
    xtick={5,10,15,20,25,30,35,40,45,50},
    legend columns=3,
  ]

\addplot[color=blue,
		only marks,
                mark=cross,
		mark=x,
                mark options={solid},
]  	
coordinates{ 
(	5	,	6.6	)
(	10	,	7.4	)
(	15	,	8.8	)
(	20	,	9.2	)
(	25	,	10	)
(	30	,	12	)
(	35	,	13	)
(	40	,	12.8	)
(	45	,	14.2	)
(	50	,	16.8	)};

  \end{axis}
\end{tikzpicture}

\begin{tikzpicture}[scale=0.51]
  \begin{axis}[ 
    xlabel={$r$},
    ylabel={Fail rate (\%)},
    xmin=5, xmax=50,
    ymin=0,
    xtick={5,10,15,20,25,30,35,40,45,50},
    legend columns=3,
  ]

\addplot[color=blue,
		only marks,
                mark=cross,
		mark=x,
                mark options={solid},
]  	
coordinates{ 
(	5	,	6.6	)
(	10	,	9.2	)
(	15	,	11.8	)
(	20	,	11	)
(	25	,	11.8	)
(	30	,	12.8	)
(	35	,	15.2	)
(	40	,	15.4	)
(	45	,	16	)
(	50	,	14.8	)};

  \end{axis}
\end{tikzpicture}
\end{minipage}
\hfill
\begin{minipage}{0.45\linewidth}
\begin{tikzpicture}[scale=0.51]
  \begin{axis}[ 
    xlabel={$r$},
    ylabel={Fail rate (\%)},
    xmin=5, xmax=50,
    ymin=0,
    xtick={5,10,15,20,25,30,35,40,45,50},
    legend columns=3,
  ]

\addplot[color=blue,
		only marks,
                mark=cross,
		mark=x,
                mark options={solid},
]  	
coordinates{ 
(	5	,	18.6	)
(	10	,	19.4	)
(	15	,	25	)
(	20	,	28	)
(	25	,	28.2	)
(	30	,	24.8	)
(	35	,	27.8	)
(	40	,	29.8	)
(	45	,	30.4	)
(	50	,	30.2	)};
  \end{axis}
\end{tikzpicture}

\begin{tikzpicture}[scale=0.51]
  \begin{axis}[ 
    xlabel={$r$},
    ylabel={Fail rate (\%)},
    xmin=5, xmax=50,
    ymin=0,
    xtick={5,10,15,20,25,30,35,40,45,50},
    legend columns=3,
  ]

\addplot[color=blue,
		only marks,
                mark=cross,
		mark=x,
                mark options={solid},
]  	
coordinates{ 
(	5	,	7.4	)
(	10	,	10	)
(	15	,	10.2	)
(	20	,	10.6	)
(	25	,	14.2	)
(	30	,	13.8	)
(	35	,	16.4	)
(	40	,	17.4	)
(	45	,	17	)
(	50	,	16.2	)};

  \end{axis}
\end{tikzpicture}
\end{minipage}

  \caption{Experiment 4}\label{4}
\end{figure}


\textbf{Summary}
From the numerical experiments, we have a better knowledge of the proposed algorithms. We prove that by employing LV randomized algorithm, we do not lose too much of the optimality compared with the provable efficient deterministic algorithm, but we dramatically reduce the time and space complexity. If the Internet Service Providers are allowed to fail to serve some flow requirements, then MC is a good choice because the running time of MC outperforms LV when the data grow large. Also, we show that the fail rate is controllable by setting large rounds of sampling. All the experiment results show the efficiency of the proposed randomized heuristic algorithms for OJPA-HS.

\section{Conclusion}
In this paper, we introduce the novel capability function in JPA-VNF to measure the potential of a virtual server for locating VNF instances and propose the OJPA-HS. OJPA-HS model allows the servers in the network to be heterogeneous, at the same time combines and generalizes many classical JPA-VNF models. However, this model is proved to be NP-hard and even no efficient deterministic online algorithm exists. Under the hardness, we propose a provable best-possible deterministic algorithm based on dynamic programming. Besides, we propose another two randomized heuristic algorithms, LV and MC, both of which conquer the shortcoming of DP for the high complexity. We show by experiments that LV performs surprisingly close with DP in terms of the output value, with much shorter running time. And a comparison between LV and MC shows that MC outperforms LV in running time for large instances. However, MC not always outputs a feasible solution, which means it is only useful when the ISP is allowed to fail to serve some requirements. But the good news is, the fail rate is controllable by carefully setting a certain parameter in MC.

\bibliographystyle{plain}

\end{document}